\documentclass[12pt,a4paper]{article}

\usepackage{color}
\usepackage{url}
\usepackage{amsmath}
\usepackage{amssymb}
\usepackage{amsthm}
\usepackage{amstext}
\usepackage{amsrefs}
\usepackage{amsfonts}
\usepackage{authblk}

\newtheorem{theorem}{Theorem}
\newtheorem{lemma}{Lemma}
\newtheorem{definition}{Definition}

\newcommand{\set}[1]{{\mathcal{#1}}}

\newcommand{\vsss}[1]{ \langle #1 \rangle}
\newcommand{\rank}[1]{ \operatorname{rank}( #1 )}
\newcommand{\diag}[1]{ \operatorname{diag}\left( #1 \right)}
\newcommand{\vs}[1]{\mathbb{#1}}
\newcommand{\vc}[1]{\mathbf{#1}}
\newcommand{\mt}[1]{\mathbf{#1}}
\newcommand{\setl}[1]{\{#1\}}
\newcommand{\di}[1]{\dim\left(#1\right)}
\newcommand{\vu}{\vs{U}}
\newcommand{\vv}{\vs{V}}
\newcommand{\vx}{\vs{X}}
\newcommand{\vy}{\vs{Y}}
\newcommand{\vz}{\vs{Z}}
\newcommand{\mta}{\mt{A}}
\newcommand{\mtb}{\mt{B}}
\newcommand{\mtm}{\mt{M}}
\newcommand{\mtx}{\mt{X}}
\newcommand{\mty}{\mt{Y}}
\newcommand{\mtz}{\mt{Z}}
\newcommand{\vcx}{\vc{x}}
\newcommand{\vcy}{\vc{y}}
\newcommand{\deq}{\triangleq}

\DeclareMathOperator{\wt}{\mathsf{w}}
\DeclareMathOperator{\supp}{\mathsf{supp}}

\begin{document}

%\author{Daniel Salmond}
%\author{Alex Grant}
%
%\author{Ian Grivell}
%\author{Terence Chan}
%
%\address[Daniel Salmond and Ian Grivell]{Cyber and Electronic Warfare Division\\ Defence Science \& Technology Organisation\\ Adelaide, South Australia}
%\email[Corresponding author]{daniel.salmond@dsto.defence.gov.au}
%
%\address[Alex Grant and Terence Chan]{Institute for Telecommunications Research\\ University of South Australia}

\author[1]{Daniel Salmond}
\author[2]{Alex Grant}
\author[1]{Ian Grivell}
\author[2]{Terence Chan}
\affil[1]{Cyber and Electronic Warfare Division, Defence Science \& Technology Organisation, Adelaide, South Australia\footnote{Correspondence to Daniel Salmond: daniel.salmond@dsto.defence.gov.au; PO Box 1500, Edinburgh SA, 5111, Australia; Tel: +61 8 7389 5000}}
\affil[2]{Institute for Telecommunications Research, University of South Australia}

\title{On the rank of random matrices over finite fields}

%\keywords{random matrices; rank; finite fields; quasi-uniform}
%
%\subjclass[2010]{Primary 15B52; Secondary 15A03, 94B05}
\date{\today}

\maketitle

\begin{abstract} A novel lower bound is introduced for the full rank probability of random finite field matrices, where a number of elements with known location are identically zero, and remaining elements are chosen independently of each other, uniformly over the field. The main ingredient is a result showing that constraining additional elements to be zero cannot result in a higher probability of full rank.  The bound then follows by ``zeroing'' elements to produce a block-diagonal matrix, whose full rank probability can be computed exactly. The bound is shown to be at least as tight and can be strictly tighter than existing bounds.  \end{abstract}

%\renewcommand{\thefootnote}{\fnsymbol{footnote}} 
%\footnotetext{\emph{Key words} whatever}     
%\renewcommand{\thefootnote}{\arabic{footnote}}

%\begin{keywords}
%random matrices; rank; finite fields; quasi-uniform
%\end{keywords}

%\begin{subjclass}
%Primary 15B52; Secondary 15A03, 94B05
%\end{subjclass}

\section{Introduction}\label{sn:beIntro}

Consider a $n\times k$ random matrix $\mtm\in\vs{F}_q^{n\times k}$ with elements $m_{ij}\in\vs{F}_q$, $i=1,2,\dots,n$, $j=1,2,\dots,k$ drawn from a $q$-ary finite field $\vs{F}_q$. Suppose that $\mtm$ is such that certain elements are zero with probability one, and the locations of these elements are known. The remaining elements are independently and identically distributed over $\vs{F}_q$. We will prove that the probability that $\mtm$ has full rank cannot increase if we fix any number of additional elements to be identically zero. This result provides a path to a novel lower bound on the probability that $\mtm$ has full rank.

% We shall use the following notation. Vectors shall follow the column convention and shall be denoted in lowercase bold-face type, e.g.~$\vc{v}$. Matrices shall be denoted uppercase, e.g.~$\mt{M}$. The element of $\mtm$ in row $i$ and column $j$ is denoted $m_{ij}$. The $j$-th column of $\mt{M}$ shall be denoted $\vc{m}_j$. Let the $n\times k$ all-one matrix be denoted $\vc{1}^{n\times k}$. Sets shall be denoted in caligraphic type, e.g.~$\set{S}$.
% Vector spaces shall be denoted in blackboard type, e.g.~$\vs{V}$. If $\set{S}$ is a set of vectors or vector subspaces, $\vsss{ \set{S} }$ will denote the smallest vector subspace spanning $\set{S}$.

% Let $\vs{U}_{\rho}\subseteq \vs{F}_q^{n\times k}$ represent a vector subspace that spans the coordinates indexed by $\rho \subseteq \setl{1,2,\dotsc,n}$. Vector subspaces of this form shall be referred to as coordinate subspaces. It follows that $\vs{U}_{\setl{1,2,\dotsc,n}} = \vs{F}_q^{n}$.

%\subsection{Problem Formulation}
The random matrix distribution of central interest in this paper is defined as follows.
\begin{definition}\label{dn:pDistUonB}
For given $\mtb\in\vs{F}_2^{n\times k}$, let $\set{U}(\mtb,q)$ be the probability distribution on random matrices $\mtm\in\vs{F}_q^{n\times k}$ whose elements  $m_{ij}$, $i=1,2,\dots,n$, $j=1,2,\dots,k$ are chosen  independently according to
\begin{equation}\label{eq:elementsOfMbyP}
 \Pr(m_{ij}=m) \deq \left\{ \begin{array}{ll}
1&\text{if $b_{ij}=0$ and $m=0$}\\
0&\text{if $b_{ij}=0$ and $m\neq 1$}\\
q^{-1}&\text{if $b_{ij} = 1$}
\end{array}\right.
\end{equation}
For matrices drawn in this way, we write $\mtm\sim \set{U}(\mtb,q)$. The corresponding support (set of matrices with positive probability) shall be denoted $\supp(\mtb,q)$.
\end{definition} 

The matrix $\mtb\in\vs{F}_2^{n\times q}$ identifies whether the elements of $\mtm\sim \set{U}(\mtb,q)$ are either uniformly distributed over $\vs{F}_q$ or zero with probability one. Let $0\leq \wt(\mtb) \leq nk$ be the number of non-zero elements of $\mtb$ (sometimes called the Hamming weight). It follows from Definition \ref{dn:pDistUonB} that for $\mtm\sim\set{U}(\mtb,q)$,
\begin{align}
\Pr(\mtm= \mta) = q^{-\wt(\mtb)}, \qquad \forall \mta \in \supp(\mtb,q).
\end{align}
Thus $\mtm$ is distributed uniformly over its support, i.e. is quasi-uniform in the terminology of Chan~\cite{ChGrBr13}.

Suppose $\mtb= \vc{1}^{n\times k}$ and hence $\wt(\mtb)=nk$. In this case, $\mtm\sim \set{U}(\mtb,q)$ shall be referred to as \emph{full weight}\footnote{Such matrices have been described as \emph{zero-free}~\cite{DuRe78} in the sense there are no identically  zero elements in $\mtm\sim \set{U}(\mtb,q)$.}. Note that $\supp(\vc{1}^{n\times k},q) = \vs{F}_q^{n\times k}$ and hence for $\mtm\sim\set{U}(\vc{1}^{n\times k},q)$, $\Pr(\mtm=\mta) = q^{-nk}$, $\forall \mta\in \vs{F}_q^{n\times k}$. Thus full-weight matrices are uniformly distributed in the usual sense.

The objective of this paper to provide a lower bound on the \emph{full rank probability}
\begin{align}
  P_{\text{FR}}(\mtm) \deq \Pr(\rank{\mtm} = \min(n,k))
\end{align}
for $\mtm\sim\set{U}(\mtb,q)$, where $\mtb\in\vs{F}_2^{n\times k}$ is given.

% This will be obtained by establishing a new result: that the full rank probability of a random matrix is monotonically non-decreasing with $\wt(\mtb)$. Lower bounds can then be obtained via a process of ``zeroing'' the elements of a random matrix to produce a block-diagonal matrix where each diagonal block is a full weight random matrix. The full rank probablility of the resulting block diagonal matrix can be computed exactly.

%\subsection{Prior Work}\label{sec:relwork}
The literature on the rank of random matrices over finite fields includes contributions from the fields of random matrix theory, combinatorics and coding theory. The contributions can be broadly summarised into two distinct approaches: a) those regarding the expected rank of matrices whose elements are identically and independently distributed over the finite field, and b) contributions regarding the expected rank for a given total matrix weight. Studholme and Blake provide a thorough survey~\cite{StBl06} across both of these approaches.

The first approach includes the work of Levitskaya and Kovalenko~\cite{KoLe93, Levitskaya05} and focuses on the asymptotic properties of the
expected rank as a function of matrix size and the distribution on the finite field elements. Similarly, Bl\"{o}mer \emph{et al.}~\cite{BlKaWe97} and Cooper~\cite{Cooper00} explored the allocation of probability mass to the zero element of the finite field, and the threshold at which the expected rank scales linearly with the matrix size.

The second approach includes Erd\H{o}s and R\'{e}nyi~\cite{ErRe63} who considered the expected rank of random binary matrices with a fixed total weight. Studholme and Blake~\cite{StBl06} presented related results regarding the minimal weight required for a random linear coding structure such that the probability of successful decoding approaches that of a full weight structure.

Lemma~\ref{lm:zeroFreeProbMat} reproduces a fundamental result for full weight matrices~\cite{BlKaWe97,Cooper00}. We will use the main idea in the proof of this well-known result (provided in the Appendix) as the basis for proving Theorem \ref{lm:lower_block_diag} in this paper. 
\begin{lemma}\label{lm:zeroFreeProbMat}
Let $\mtm\sim\set{U}(\vc{1}^{n\times k},q)$,  where $n\geq k$. Then, 
\begin{align}\label{eq:equiprob}
P_{\text{FR}}(\mtm) = \prod_{i=n-k+1}^n\left( 1- \frac{1}{q^{i}}\right)
\end{align}
\end{lemma}

A lower bound on the full rank probability of any square random matrix over a finite field was developed by Ho \emph{et al.}\footnote{The bound of Ho \emph{et al.} was intended for the analysis of Edmonds matrices, but also applies to square random matrices}~\cite{HoMeKoKaEfShLe06}. Their bound assumes the existence of at least one full rank realisation. Noting that the determinant of a square random matrix is a multivariate polynomial, their bound follows from repeated application of the Schwarz-Zippel lemma~\cite{HoJaVyXi11,Moshkovitz10,Schwartz80}.
Let $\mtm\sim\set{U}(\mtb,q)$, where $\mtb\in\vs{F}_2^{n\times n}$ is such that $\Pr(\det(\mtm) = 0) \neq 1$. Then the bound of Ho \emph{et al.} is \begin{align}\label{eq:hobound1}
P_{\text{FR}}(\mtm) \geq \left(1-\frac{1}{q} \right)^{n}
\end{align} 
Apart from  the condition $\Pr(\det(\mtm) = 0) \neq 1$, this bound is independent of $\mtb$. Furthermore, the bound is tight, with equality achieved when $\mtb=\mt{I}_n$, the $n\times n$ identify matrix. In this case, $\mtm\sim\set{U}(\mt{I}_n,q)$ has full rank if and only if $m_{ii}\neq 0$, $i=1,2,\dots,n$.

The bound that we provide in this paper exhibits a greater dependence on the structure of $\mtb$, and is at least as tight as and can be strictly tighter than  (\ref{eq:hobound1}).

\section{Main Results}\label{sn:beMain}
Define a partial order $\prec$ on binary matrices as follows. For $\mta\neq\mtb\in\vs{F}_2^{n\times k}$, $\mta \prec \mtb$ if $b_{ij}=0\implies a_{ij} =0$, $\forall 1\leq i\leq n, 1\leq j\leq k$. 

The following two theorems are the main results of the paper.  Proofs are given in the Appendix.
\begin{theorem}\label{lm:lower_block_diag}
Let $\mtx\sim\set{U}(\mta,q)$ and $\mty\sim\set{U}(\mtb,q)$ be independent random matrices where $\mta\prec\mtb \in \vs{F}_2^{n\times k}$. Then
%Let $\mtb^*,\mtb \in \set{B}^{n\times k}$ such that $n\leq k$ and $\mtb^*< \mtb$.  Let $\mtm^*\sim\set{U}(\mtb^*,q)$ and $\mtm\sim\set{U}(\mtb,q)$. Then,
\begin{align}
P_{\text{FR}}(\mtx) \leq P_{\text{FR}}(\mty)
\end{align}
\end{theorem}
%One application of Theorem \ref{lm:lower_block_diag} is a lower bound on the full rank probability of a random matrix $\mtm\sim\set{U}(\mtb,q)$ where $\mtb\in \set{B}^{n\times k}$. 
\begin{theorem}\label{th:mainbound}
Let $\mta = \diag{ \mta_1,\dotsc, \mta_L } \prec \mtb \in \vs{F}_2^{n\times k}$, where $\mta_\ell=\vc{1}^{n_\ell\times k_\ell}$ and $n_\ell\leq k_\ell$, $\ell= 1,2,\dots,L$.  Let $\mtx\sim \set{U}(\mta, q)$ and $\mty\sim\set{U}(\mtb,q)$.
Then,
\begin{align}
  P_{\text{FR}}(\mty)  &\geq P_{\text{FR}}(\mtx) \\
  &= \prod_{\ell=1}^{L} \prod_{i=k_\ell-n_\ell+1}^{k_\ell} \left(1-\frac{1}{q^i}\right)
  \label{eq:mainbound} \\
&\geq \left(1-\frac{1}{q}\right)^n.\label{eq:tight}
\end{align}
\end{theorem}
Theorem \ref{th:mainbound} suggests an approach to obtain lower bounds on the full rank probability of $\mtm\sim \set{U}(\mtb,q)$ for any $\mtb\in \vs{F}_2^{n\times k}$. Let \emph{zeroing} mean the act of setting a matrix element to zero. Noting rank is invariant under row and column permutations: 
\begin{enumerate}
\item Process $\mtb$ by applying any combination of row and column permutations and zeroing elements to obtain $\mta=\diag{\mta_1,\dots,\mta_L}$, with $\mta_\ell=\vc{1}^{n_\ell\times k_\ell}$, where $n_\ell\leq k_\ell$ and $\sum_{\ell=1}^L n_\ell=n$.  
\item Apply Theorem~\ref{th:mainbound}, computing (\ref{eq:mainbound}) for the $n_\ell$ and $k_\ell$ obtained in Step 1. 
\end{enumerate}
The bound can clearly be optimised over the choices made in Step 1. This is discussed in Section~\ref{sn:beDiscussion}. 

For given $\mtb\in \vs{F}_2^{n\times k}$, the bound~(\ref{eq:mainbound}) obtained in Step~2 is a product of $n$ factors of the form $1-q^{-a_i}$, $i=1,2,\dots,n$, where the exponents $a_i\geq 1$ depend on the structure of $\mtb$ via the dimensions of the $\mta_\ell$ obtained in Step~1. As a consquence of~\eqref{eq:tight}, the proposed bound is as tight as, and can be strictly tighter than~(\ref{eq:hobound1}).

%, since for $q,n\geq 1$ and $a_i\geq 1$, $i=1,2,\dots,n$ we have 
%\begin{lemma} Let $q,n\geq 1$ and $a_i\geq 1$, $i=1,2,\dots,n$ Then
%  Let $\mtm\sim\set{U}(\mtb,q)$ where $\mtb\in \vs{F}_2^{n\times n}$ such that $\Pr(\det(\mtm) =  0)\neq 1$. Then, 
%  \begin{equation}
 %   \prod_{i=1}^n \left(1-\frac{1}{q^{a_i}} \right) \geq \left(1-\frac{1}{q}\right)^n.
 % \end{equation}
%This is a direct consequence of $q^{-a} \leq q^{-1}$ for $a\geq 1$, with equality if and only if $a=1$.
%\end{lemma}

%\begin{proof} A direct consequence of $q^{-a} \leq q^{-1}$ for $a\geq 1$, with equality if and only if $a=1$.
% \begin{align}
% \frac{1}{q^a}&\leq \frac{1}{q},\qquad \forall a\in \setl{1,2,\dots}\\
% \Rightarrow \left(1-\frac{1}{q}\right) &\leq \left(1-\frac{1}{q^a}\right), \qquad\forall a\in\setl{1,2,\dots}\\
% \Rightarrow \left(1-\frac{1}{q}\right)^n &\leq \prod_{i=1}^n\left(1-\frac{1}{q^{a_i}}\right)\\
% \Rightarrow \left(1-\frac{1}{q}\right)^n &\leq P_{\text{FR}}(\mtm)
% \end{align} 
%\end{proof}

Finally, we can also obtain an upper bound.
\begin{theorem}
  Let $\mtm\sim\set{U}(\mtb,q)$ where $\mtb\in \vs{F}_2^{n\times k}$ and $n\leq k$. Then, 
\begin{equation}
  P_{\text{FR}}(\mtm)\leq \prod_{i=k-n+1}^k\left( 1- \frac{1}{q^{i}}\right),
\end{equation}
with equality when $\mtb=\mt{1}^{n\times k}$.
\end{theorem}
\begin{proof}
Apply Theorem~\ref{lm:lower_block_diag} and Lemma~\ref{lm:zeroFreeProbMat}, noting $\mtb \prec \mt{1}^{n\times k}$. 
\end{proof}

\section{Discussion}\label{sn:beDiscussion}
The proposed bound on full rank probability is applicable to situations involving ``structured'' random matrices $\mtm\sim\set{U}(\mtb,q)$ whose non-identically-zero elements are drawn uniformly from a finite field.  For example, in multiplicative matrix operator communications channels~\cite{NoUcSi11,YaHoMeYe11}, information is encoded as a vector of finite field elements and the channel action is represented by a random matrix. The ability of the receiver to decode depends on whether this matrix has full rank. In the case where the channel is a network performing random linear network coding~\cite{HoMeKoKaEfShLe06}, the structure of this channel matrix is determined by the network topology.  The application of the bound of Thereom \ref{th:mainbound} to evaluate the performance random linear network codes is considered in~\cite{Salmond12}.

Practical implementation and optimisation of Step~1 of the algorithm described in Section~\ref{sn:beMain} requires further consideration. The challenge is to identify the elements of the random matrix that can be zeroed, while maximising the the full rank probability of the resulting block-diagonal matrix. There is a tradeoff between the complexity of the algorithm and tightness of the resulting bound.

Two greedy search algorithms were proposed in~\cite{Salmond12} to perform the required block-diagonalisation. Both algorithms treat $\mtb$ as the incidence matrix of a bipartite graph. The identification of bicliques is the key mechanic of these algorithms, where a biclique is defined as a set of left and right nodes in the bipartite graph with the property that every left node of the biclique is connected to every right node of the biclique. 

In the first algorithm a biclique is chosen and then grown by identifying additional left and right nodes that increase the size of the biclique whilst preserving the biclique property. When the biclique is maximal, i.e. can no longer increase in size, the additional links between the biclique and remaining nodes are removed  from the bipartite graph, which is equivalent to zeroing the corresponding elements of the incidence matrix. Care is taken to ensure that the removal of any link does not produce a random matrix which has zero full rank probability. The algorithm iterates by choosing a new biclique in the remaining graph. The algorithm completes when all bicliques have been grown to maximal size. This algorithm is  $O(n^2)$.

The second algorithm operates similarly except that multiple bicliques are considered in parrallel. During each iteration, multiple links are identified as candidates for removal, which ultimately affects the realisation of maximal bicliques. At the end of an iteration, the algorithm removes only the candidate link that maximises the resulting full rank probability. The second algorithm has higher complexity than the first, $O(n^4)$, but produces block-diagonal matrices with much higher full-rank probabilities and hence tighter bounds.

\appendix
\section*{Appendix}

\begin{proof}[Proof of Lemma \ref{lm:zeroFreeProbMat}]
Let $E_j$ denote the event that the first $j$ columns of $\mtm$ are linearly independent. Then,
\begin{align}
  \Pr(E_{j+1}) &= \Pr(E_{j+1}\mid E_{j})P(E_j) +
  {\Pr(E_{j+1}\mid\bar{E}_{j})} P(\bar{E}_j)\\ 
  &\stackrel{(a)}{=} \Pr(E_{j+1}\mid E_{j})P(E_j)\\
  &\stackrel{(b)}{=} \frac{q^n-q^j}{q^n}P(E_j)\label{eq:zfpm1}
\end{align}
where $(a)$ is due to $\Pr(E_{j+1} \mid \bar{E}_{j})=0$, and $(b)$ is because under $E_j$, the first $j$ columns span a vector space of dimension $j$, with volume $q^j$. Therefore, there are $q^n-q^j$ equiprobable realisations of column $j+1$ that are linearly independent of the first $j$ columns. 

Event $E_1$ occurs provided that the first column in not all-zero, hence $\Pr(E_1) = \frac{q^n-1}{q^n}$ and
\begin{align}
\Pr(E_k) &=  \Pr(E_1)\prod_{j=2}^k \Pr(E_j\mid E_{j-1})\\
&=\prod_{j=1}^k\frac{q^n-q^{j-1}}{q^n}\\
%&=  \prod_{j=1}^k\left( 1- \frac{1}{q^{n-j}}\right)\\
&=  \prod_{i=n-k+1}^n\left( 1- \frac{1}{q^{i}}\right)
\end{align}
\end{proof}

The proof of Theorem \ref{lm:lower_block_diag} requires the following lemma, which also appears as Lemma 1 of~\cite{Bu80}.
\begin{lemma}\label{lm:bndsOnDims}
  Let $\vv \subseteq \vs{F}_q^n$, and let $\vu_\alpha\subseteq\vu_\beta\subseteq\vs{F}_q^n$ be coordinate subspaces spanning the dimensions indexed by $\alpha\subseteq\beta\subseteq\{1,2,\dots,n\}$ such that $\di{\vu_\alpha}=|\alpha|$ and $\di{\vu_\beta}=|\beta|$. Then 
  \begin{align}
    0\leq \di{\vv\cap \vu_{\beta}} - \di{\vv\cap \vu_{\alpha}}\leq |\beta|-|\alpha|.
\end{align}
\end{lemma}
\begin{proof}
  The lower bound follows from the fact that $\alpha\subseteq\beta$ implies $\vu_\alpha\subseteq\vu_\beta$ and hence $\vv\cap \vu_{\alpha} \subseteq \vv\cap \vu_{\beta}$. Applying a well-known identity~\cite[Theorem 27.15]{Warner90}.
\begin{align}
\di{\vv\cap \vu_{\alpha}} &= \di{(\vv\cap \vu_{\beta})\cap \vu_{\alpha}}\\
&= \di{\vv\cap\vu_{\beta}}+\di{\vu_{\alpha}} - \di{\vsss{\vv\cap\vu_{\beta},\vu_{\alpha}}}
\end{align}
Since $\vv \cap \vu_{\beta}$ and $\vu_{\alpha}$ are subsets of $\vu_{\beta}$ it follows that $\di{\vsss{\vv\cap\vu_{\beta},\vu_{\alpha}}}\leq |\beta|$. This, together with $\di{\vu_\alpha} = |\alpha|$ delivers the upper bound.
\end{proof}

\begin{proof}[Proof of Theorem \ref{lm:lower_block_diag}]
The main idea is that zeroing one element of $\mtb$ reduces $|\supp(\mtb,q)|$ by a factor of $q$ whereas the number of full rank realisations of $\mtm\sim\set{U}(\mtb,q)$ is reduced by a factor of $q$ or more. 

To this end, let $\mta\prec\mtb\in \vs{F}_2^{n\times k}$ such that $a_{nk} = 0$ and $b_{nk} = 1$ and $a_{ij} = b_{ij}$, $(i,j)\neq (n,k)$. Thus $\mta$ differs from $\mtb$ only in element $nk$, which has been zeroed. As rank is invariant to row and column permutations, the following result  holds without loss of generality for  $\mta\prec\mtb$ differing in a single element. Without loss of generality, we also assume $n\leq k$.

Let $\mtx\sim\set{U}(\mta,q)$ and $\mty\sim\set{U}(\mtb,q)$. Then there are a factor of $q$ fewer realisations of $\mtx$ than $\mty$. 
Partition $\mtx = [\tilde\mtx \mid \vcx]$, where $\tilde\mtx\in\vs{F}_q^{n\times (k-1)}$ is the first $k-1$ columns of $\mtx$ and $\vcx\in\vx$ is the $k$-th column, which lies in the coordinate subspace $\vx$ spanning the dimensions indexed by the non-zero entries of column $k$ of $\mta$. Similarly define $\tilde\mty$ and $\vcy\in\vy$.
It follows that $\vx\subset\vy$, and $\di{\vx} = w-1$ where $w=\di{\vy}$. 

The set of all realisations of $\tilde\mtx$ and $\tilde\mty$ are identical, as $\mta$ and $\mtb$ agree in the first $k-1$ columns. Take any such realisation, $\mtz\in\vs{F}_q^{n\times(k-1)}$ with column space $\vz$. There are three cases to consider: (a) $\rank{\mtz}=n$, (b) $\rank{\mtz}=n-1$ and (c) $\rank{\mtz}<n-1$. 

\renewcommand{\theenumi}{\alph{enumi}}
\begin{enumerate}
\item $\rank{\mtz}=n$

  In this case,  $\vcx,\vcy\in\vz$. All realisations of $\mtx = [\mtz \mid \vcx]$ and $\mty = [\mtz \mid \vcy]$ are full rank. There are a factor $q$ fewer full rank realisations of $\mtx$ compared to $\mty$.

\item  $\rank{\mtz}=n-1$

For $\mtx$ and $\mty$ to have full rank we must have $\vcx,\vcy\not\in\vz$. We proceed following the main idea in the proof of Lemma \ref{lm:zeroFreeProbMat}.

Let $d' = \di{\vz\cap\vx}$, and $d = \di{\vz\cap\vy}$. By Lemma \ref{lm:bndsOnDims},
\begin{align}\label{eq:d_constraint}
d'\leq d\leq d'+1
\end{align}
hence there are two subcases to consider: $d=d'$ and $d=d'+1$.
\renewcommand{\theenumii}{\alph{enumi}.\arabic{enumii}}
\begin{enumerate}
\item $d = d'$

The linearly independent realisations of $\vcx$ are those vectors in $\vx$ lying  outside of $\vz$. There are
\begin{align} 
  |\vx|-|\vz\cap\vx|=q^{w-1}-q^d = q^{d}(q^{w-d-1}-1)
\end{align} 
such vectors. Similarly, the  number of linearly independent choices for $\vcy\in\vy$ is 
\begin{align} |\vy|-|\vz\cap\vy|=q^w-q^d = q^{d}(q^{w-d}-1).
\end{align}
In the event that $\vx\subseteq \vz$ then $w-1 = d'=d$ and there are no linearly independent realisations $\vcx\in\vx$. However in this case $d=d'\implies \vy\cap\vz=\vx$ and there are $q^w-q^d=q^w-q^{w-1}=q^{w-1}(q-1)\geq 1$ independent realisations of $\vcy\in\vy$. If instead $w-1>d$ it follows that 
\begin{align}
\frac{|\vy|-|\vz\cap\vy|}{|\vx|-|\vz\cap\vx|}&=\frac{q^{w-d}-1}{q^{w-d-1}-1}\\
%&= \frac{(q^{w-d-1}-1)q+q-1}{q^{w-d-1}-1}\\
&= q + \frac{q-1}{q^{w-d-1}-1} \label{eq:ratioOfqVals}\\
&> q
\end{align}
Hence the number of realisations of $\vcx\in\vx$ that are linearly independent of $\vz$  is reduced by a factor greater than $q$ compared to $\vcy\in\vy$.

\item  $d=d'+1$

In this case
\begin{equation*}
\frac{|\vy|-|\vz\cap\vy|}{|\vx|-|\vz\cap\vx|}=\frac{q^{w}-q^d}{q^{w-1}-q^{d-1}} = q
%&= \frac{q^d(q^{w-d}-1)}{q^{d-1}(q^{w-d}-1)}\\
%&= q
\end{equation*}
and the reduction in full rank realisations is by a factor $q$.
\end{enumerate}

\item  $\rank{\mtz}<n-1$

There are no choices resulting in $\mtx$ or $\mty$ being full rank.
\end{enumerate}

In summary, there are a factor of $q$ fewer realisations of $\mtx$ compared to $\mty$, whereas the number of full rank realisations of $\mtx$ is reduced by a factor of at least $q$. The full rank probability of $\mtx$ is as follows:
\begin{align}
P_{\text{FR}}(\mtx) &= \frac{|\setl{\hat{\mtx}\in \supp(\mta,q): \rank{\hat{\mtx}} = n}| }{|\supp(\mta,q)|}\\
&=\frac{|\setl{\hat{\mtx}\in \supp(\mta,q): \rank{\hat{\mtx}} = n}| }{\frac{1}{q}|\supp(\mtb,q)|}\\
&\leq \frac{\frac{1}{q}|\setl{\hat{\mtx}\in \supp(\mtb,q): \rank{\hat{\mtx}} = n}| }{\frac{1}{q}|\supp(\mtb,q)|}\label{eq:ineq_due_to_removal}\\
&= P_{\text{FR}}(\mty)
\end{align}

% Let $\mtb''\in \set{B}^{n\times k}$ such that $\mtb''<\mtb'$ where $\mtb''$ and $\mtb'$ differ in one element. Let $\mtm''\sim\set{U}(\mtb'',q)$. Then, by the same reasoning,
% \begin{align}
% P_{\text{FR}}(\mtm'')\leq P_{\text{FR}}(\mtm')\leq P_{\text{FR}}(\mtm)
% \end{align}
Which establishes the result for matrices that differ in a single element. For arbitrary $\mta\prec\mtb$ differing in more than one element, we can zero one element at a time to achieve a chain of inequalities that provides the general result.

% It follows that if $\mtb^*<\mtb$ then a chain of matrices, $\mtb^{(0)},\mtb^{(1)}, \mtb^{(2)},\dotsc,\mtb^{(N)}\in \set{B}^{n\times k}$ can be chosen such that $\mtb^*=\mtb^{(0)}$, $\mtb=\mtb^{(N)}$ and $\mtb^{(i)}<\mtb^{(i+1)}$ where $\mtb^{(i)}$ and $\mtb^{(i+1)}$  differ by one element. Consequently,
% \begin{align}
% P_{\text{FR}}(\mtm^*)\leq P_{\text{FR}}(\mtm)
% \end{align}
\end{proof}

\begin{proof}[Proof of Theorem \ref{th:mainbound}]
From Theorem \ref{lm:lower_block_diag} we have
\begin{align}
P_{\text{FR}}(\mty)&\geq P_{\text{FR}}(\mtx)\\
&=\prod_{\ell=1}^L P_{\text{FR}}(\mtx_\ell)
\end{align}
where the second line follows from the fact that $\mtx=\diag{\mtx_1,\mtx_2,\dots,\mtx_L}$ is block diagonal and can only be full rank if each independently chosen block is full rank. Applying \eqref{eq:equiprob} to each of these full weight blocks $\mtx_\ell\sim\set{U}(\mta_\ell,q)$ completes the proof.
\end{proof}

\bibliography{refs}

\end{document}